\RequirePackage{fix-cm}
\documentclass[a4paper]{jpconf} 
\usepackage[utf8]{inputenc}
\usepackage{amssymb}
\usepackage{amsmath}
\usepackage{amsthm}
\usepackage{graphicx}
\usepackage{mathtools}
\usepackage{amsfonts}
\usepackage{setspace}
\usepackage{mathrsfs}
\usepackage[usenames,dvipsnames]{xcolor}
\usepackage[pageanchor=false,colorlinks=true,linkcolor=blue,citecolor=blue,urlcolor=blue]{hyperref}
\usepackage{tikz}
\usepackage{tikz-3dplot}
\usepackage{pgfplots}
\usetikzlibrary{calc,intersections,positioning,angles,quotes}
\usetikzlibrary{3d,shapes}
\pgfplotsset{compat=1.11}

\newcommand\bR{\mathbb{R}}

\newcommand\sM{\mathcal{M}}
\newcommand\sH{\mathcal{H}}
\newcommand\sA{\mathcal{A}}
\newcommand\sC{\mathcal{C}}

\newcommand\sT{\mathcal{T}}
\newcommand\mM{M}

\newcommand{\tspan}{\mathrm{span}}

 \newcommand{\vol}{\mathrm{vol}}
 \newcommand{\cone}{\mathrm{cone}}
 \newcommand{\ls}{\langle}
 \newcommand{\rs}{\rangle}
   \newcommand{\re}{\mathcal{R}e}

\newcommand{\defin}{:=}
\newtheorem{theorem}{Theorem}[section]

\newtheorem{proposition}{Proposition}[section]

\newtheorem{lemma}{Lemma}[section]
\newtheorem{definition}{Definition}[section]
\newtheorem{remark}{Remark}[section]

\begin{document}
\date{}

\title{Optimality of 1-norm regularization among weighted 1-norms for sparse recovery: a case study on how to find optimal regularizations}
 \author{Yann  Traonmilin$^a$ and Samuel Vaiter$^b$}
 \address{$^a$CNRS, Institut de Mathématiques de Bordeaux, 33400 Talence, France; $^b$CNRS, Institut de Mathématiques de Bourgogne, 9 avenue Alain Savary, 21000, Dijon}

\ead{$^a$yann.traonmilin@math.u-bordeaux.fr, $^b$samuel.vaiter@u-bourgogne.fr}


\begin{abstract}
The 1-norm was proven to be a good convex regularizer for the recovery of sparse vectors from under-determined linear measurements. It has been shown that with an appropriate measurement operator, a number of measurements of the order of the sparsity of the signal (up to log factors) is sufficient for stable and robust recovery. More recently, it has been shown that such recovery results can be generalized to more general low-dimensional model sets and (convex) regularizers. These results lead to the following question: to recover a given low-dimensional model set from linear measurements, what is the ``best'' convex regularizer? To approach this problem, we propose a general framework to define several notions of ``best regularizer'' with respect to a low-dimensional model. We show in the minimal case of sparse recovery in dimension 3 that the 1-norm is optimal for these notions. However, generalization of such results to the  n-dimensional case seems out of reach. To tackle this problem, we propose looser notions of best regularizer and show that the 1-norm is optimal among weighted 1-norms for sparse recovery within this framework.
\end{abstract}
\vspace*{-2em}
\section{Introduction}
We consider the observation model in a Hilbert space $\sH$ (with associated norm $\|\cdot\|_\sH$):
\begin{equation}
 y = Mx_0
\end{equation}
where $M$ is an under-determined linear operator,  $y$ is a $m$-dimensional vector and $x_0$ is the unknown. We suppose that $x_0$ belongs to a low-dimensional model $\Sigma$ (a union of subspaces). We consider the following minimization program.
\begin{equation} \label{eq:minimization1}
 x^* \in \arg \min_{Mx=y} R(x)
\end{equation}
where $R$ is a regularization function. A huge body of work gives practical regularizers ensuring that $x^* = x_0$ for several low-dimensional models (in particular sparse and low rank models, see \cite{Foucart_2013} for a most complete review of these results). The operator $M$ is generally required to satisfy some property (e.g., the restricted isometry property) to guarantee recovery.  In this work, we aim at finding the ``best'' regularizer for exact recovery of $x_0 \in \Sigma$.  Ideally we would like to set $R = \iota_\Sigma$ (the characteristic function of $\Sigma$) but it is not practical in many cases (sparse and low rank recovery) as it is generally not convex, and even NP-hard to compute as a combinatorial problem. Consequently, we  restrict the search for the best regularizer to a class of interesting regularizers $\sC$. In  our examples, the set $\sC$ is a subset of the set of convex functions.  Other interesting classes might be considered, such as partly smooth functions~\cite{Vaiter2015Modelselectionlow}. 

\subsection{Best regularizer with respect to a low dimensional model}
Defining what is the ``best'' regularizer in $\sC$ for recovery is not immediate. Ideally, to fit to the inverse problem, we must define a compliance measure that depends on both the kind of unknown and measurement operator we consider.  If we have some knowledge that $M \in \sM$ where $\sM$ is a set of linear operators, we want to define a compliance measure $A_{\Sigma,\sM}(R)$ that tells us if a regularizer is good in these situations, and maximize it.  Such maximization might yield a function $R^*$ that depends on $M$ (e.g., in \cite{Soubies_2015}, when looking for tight continuous relaxation of the $\ell^0$ penalty a dependency on $M$ appears). We aim for a more \emph{universal} notion of optimal regularizer that does not depend on $M$. Hence, we look for a compliance measure  $A_{\Sigma}(R)$ and its maximization
\begin{equation}
R^* \in \arg\max_{R \in \sC} A_{\Sigma}(R). 
\end{equation}
In the sparse recovery example studied in this article, the existence of a maximum of $A_{\Sigma}(R)$ is verified. However, we could ask ourselves what conditions on $A_\Sigma(R)$ and $\sC$ are necessary and sufficient for the existence of a maximum, which is out of the scope of this article.

\subsection{Compliance measures}
When studying recovery with a regularization function $R$, two types of guarantees are generally used: uniform and non-uniform. To describe these recovery guarantees, we use the following definition of descent vectors. 	

 \begin{definition}[Descent vectors]
 For any $x \in \sH$, the collection of descent vectors of $R$ at $x$ is
 \begin{equation}
 \sT_{R}(x) \defin \left\{ z \in \sH : R(x+z) \leq R(x) \right\}.
 \end{equation}
\end{definition}
We write $\sT_R(\Sigma):=\bigcup_{x\in \Sigma} \sT_R(x)$. Recovery is characterized by descent vectors (recall that $x^*$ is the result of minimization~\eqref{eq:minimization1}): 
\begin{itemize}
 \item Uniform recovery: Let $M$ a linear operator. Then  ``for all $x_0 \in \Sigma$, $x^* =x_0$'' is equivalent to $\sT_R(\Sigma) \cap \ker M =\{0\}$.
 \item Non-uniform recovery: Let $M$ a linear operator and $x_0 \in \Sigma$.  Then $x^* =x_0$  is equivalent to $\sT_R(x_0) \cap \ker M = \{0\}$.
\end{itemize}

Hence, a regularization function $R$ is ``good'' if  $\sT_R(\Sigma)$ leaves a lot of space for $\ker M$ to not intersect it (trivially). In dimension $n$, if there is no orientation prior on the kernel of $M$, the amount of space left can be quantified by the ``volume'' of $\sT_R(\Sigma) \cap S(1)$ where $S(1)$ is the unit sphere with respect to $\|\cdot\|_\sH$. Hence a compliance measure for uniform recovery can be defined as 
\begin{equation}
A_\Sigma^U(R) := 1 - \frac{\vol\left(\sT_R(\Sigma) \cap S(1)\right)}{\vol(S(1))}.
\end{equation}
More precisely, here, the volume $\vol(E)$ of a set $E$ is the measure of $E$ with respect to the uniform measure on the sphere $S(1)$ (i.e. the $n-1$-dimensional Haussdorf measure of $\sT_R(\Sigma) \cap S(1)$). When maximizing this compliance measure with convex regularizers (proper, coercive and continuous), it has been shown that we can limit ourselves to atomic norms with atoms included in the model \cite[Lemma 2.1]{Traonmilin_2016}. When looking at non-uniform recovery for random Gaussian measurements, the quantity $\frac{\vol\left(\sT_R(x_0) \cap S(1)\right)}{\vol(S(1))}$ represents the probability that a randomly oriented kernel of dimension 1 intersects (non trivially) $\sT_R(x_0)$. The highest probability of intersection with respect to $x_0$  quantifies the lack of compliance of $R$, hence we can define: 
\begin{equation}\label{eq:anusr}
A_\Sigma^{NU}(R) := 1 - \sup_{x \in \Sigma} \frac{\vol\left(\sT_R(x) \cap S(1)\right)}{\vol(S(1))}
\end{equation}
Note that this can be linked with the Gaussian width and statistical dimension theory of sparse recovery \cite{Chandrasekaran_2012, Amelunxen_2014}.

\begin{remark} In infinite dimension, the volume of the sphere $S(1)$ vanishes, making the measures above uninteresting. However, \cite{Traonmilin_2016} and \cite{Puy_2015} show that we can often come back to a low-dimensional recovery problem in an intermediate finite (potentially high dimensional) subspace of $\sH$. Adapting the definition of $S(1)$ to this subspace allows to extend these compliance measures. 
\end{remark}

Another possibility for the $n$-dimensional case, which we develop in this article, is to use recovery results based on the restricted isometry property (RIP). They have been shown to be adequate for multiple models~\cite{Traonmilin_2016}, to be tight in some sense~\cite{Davies_2009} for sparse and low rank recovery, to be necessary in some sense~\cite{Bourrier_2014} and to be well adapted to the study of random operators \cite{Puy_2015}. In particular, it has been shown that if $M$ has a RIP with constant $\delta< \delta_\Sigma(R)$ on the secant set $\Sigma-\Sigma$, with $\delta_\Sigma(R)$ being fully determined by $\Sigma$ and $R$~\cite{Traonmilin_2016}, then stable recovery is possible. Hence, by taking 

\begin{equation}
A_\Sigma^{RIP}(R) := \delta_\Sigma(R),
\end{equation}
the larger $A_\Sigma^{RIP}$ is, the less stringent are RIP recovery conditions for recovery of elements of $\Sigma$ with $R$. We develop these ideas more precisely in Section~\ref{sec:RIP_compliance} and perform the maximization with such measure in the sparse recovery case.

\subsection{Optimality results} 
Within this framework we show the following results:
\begin{enumerate}
 \item In Section~\ref{sec:3d_example}, when $\sH= \bR^3$, $\Sigma = \Sigma_1$ the set of $1$-sparse vectors and $\sC$ is the set of (symmetric) atomic norms with atoms included in the model. Both $A_\Sigma^U(R)$ and $A_\Sigma^{NU}(R)$ are uniquely maximized by multiples of the $\ell^1$-norm (Theorem~\ref{th:3d-u-nu}). In this case, it is possible to exactly compute and maximize the compliance measure. While this study gives a good geometrical insight of the quantities at hand, extending these exact calculations to the general $k$-sparse recovery in dimension $n$ seems out of reach. 
 \item In Section~\ref{sec:RIP_compliance}, we describe precisely how compliance measures based on RIP recovery can be defined. We then study two of these measures (based on state of the art recovery guarantees) for $\sH= \bR^n$, $\Sigma = \Sigma_k$ the set of $k$-sparse vectors (vectors with at most $k$ non zero coordinates) and $\sC$ the set of weighted $\ell^1$-norms. We show that both of these measures are uniquely maximized by multiples of the $\ell^1$-norm (Theorem~\ref{th:RIP_nec} and Theorem~\ref{th:RIP_suff}). 
\end{enumerate}

\section{The case of 1-sparse vectors in 3D} \label{sec:3d_example}

We investigate in detail the case of 3-dimensional vectors which are 1 sparse, i.e. the simplest interesting case (uniform recovery of 1-sparse vectors in 2 dimensions is impossible if $M$ is not invertible).
Here, we have $\Sigma= \Sigma_1$ and $\sH = \bR^3$.
We consider weighted $\ell^1$-norms, i.e., atomic norms of the form $\sC =\{ \|\cdot\|_w : \|z\|_w = w_1|z_1|+w_2|z_2|+w_3|z_3| \} = \{ \|\cdot\|_\sA  : \sA = \{ \pm w_i^{-1} e_i\}_{i=1,3}, w_i > 0, \max |w_i| =1\}$ (where  $\|\cdot\|_\sA$ is the atomic norm generated by atoms $\sA$ ). 
To maximize $A_\Sigma^{U}$ (respectively $A_\Sigma^{NU}$), we have to compute  the surface of the intersections of 3 descent cones with the unit sphere $S(1)$ (respectively the surface of the biggest possible intersection of a descent cone with  $S(1)$).

\begin{figure}[!t]
  \centering
%
%
%
%

\includegraphics[width=0.3\linewidth]{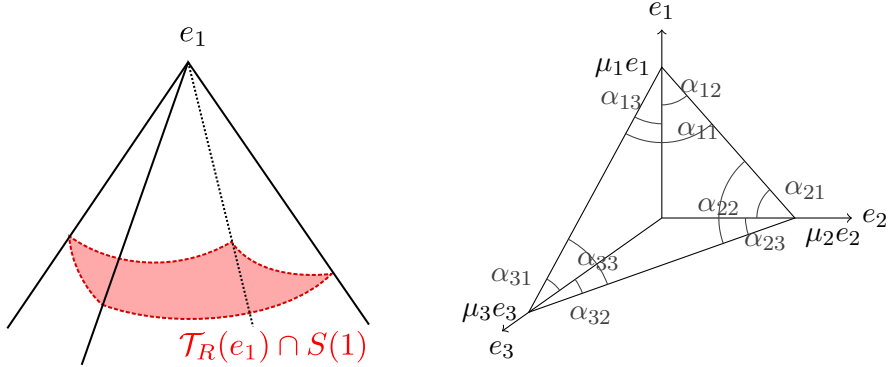}
\hspace{8mm}
\begin{tikzpicture}[scale=2.5]
    \draw[->] (0,0) -- (0,1) node[above] {$e_1$};
    \draw[->] (0,0) -- (1,0) node[right] {$e_2$};
    \draw[->] (0,0) -- (-0.84,-0.6) node[below] {$e_3$};

    \coordinate (orig) at (0,0);
    \coordinate [label=left:$\mu_1 e_1$] (w1) at (0,0.8);
    \coordinate [label=below right:$\mu_2 e_2$] (w2) at (0.7,0);
    \coordinate [label=left:$\mu_3 e_3$] (w3) at (-0.7,-0.5);

    \draw (w1) -- (w2) -- (w3) -- cycle
    pic[draw,angle radius=0.5cm,color=darkgray,"$\alpha_{12}$",right] {angle=orig--w1--w2}
    pic[draw,angle radius=0.75cm,,color=darkgray,"$\alpha_{13}$",left] {angle=w3--w1--orig}
    pic[draw,angle radius=1cm,,color=darkgray,"$\alpha_{11}$",below right] {angle=w3--w1--w2}
    pic[draw,angle radius=0.5cm,color=darkgray,"$\alpha_{21}$",above right] {angle=w1--w2--orig}
    pic[draw,angle radius=0.65cm,color=darkgray,"$\alpha_{23}$",below] {angle=orig--w2--w3}
    pic[draw,angle radius=1cm,color=darkgray,"$\alpha_{22}$",left] {angle=w1--w2--w3}
    pic[draw,angle radius=0.5cm,color=darkgray,"$\alpha_{31}$",above left] {angle=orig--w3--w1}
    pic[draw,angle radius=0.75cm,color=darkgray,"$\alpha_{32}$",below right] {angle=w2--w3--orig}
    pic[draw,angle radius=1.1cm,color=darkgray,"$\alpha_{33}$",above right] {angle=w2--w3--w1};
  \end{tikzpicture}
  \caption{Left: a representation of the surface of interest. Right: parameterization of the angles of a tetrahedron.\vspace*{-5mm}
}
  \label{fig:3d-1-sparse}
\end{figure}
As shown in Figure~\ref{fig:3d-1-sparse}, with symmetries, we need  to compute the intersection of 3 descent cones of a weighted $\ell^1$-norm at 1-sparse vectors. This is the object of Lemma~\ref{lem:3d-u-nu-descent-vol}. Note that in the case of asymmetrical atoms, we would then need to calculate $6 \times 4$ intersections of the sphere with a tetrahedron. 
To ease the notations, we introduce the following quantities which represent the cosine of the angles of a tetrahedron of interest.
For $i \neq j$, set $\mu_i = w_i^{-1}$,$  \beta_{ij}(\mu) = \cos(\alpha_{ij})
  = \left( 1 + \left(\mu_j/\mu_i\right)^2 \right)^{-\frac{1}{2}}
$ and $\beta_{ii} = \prod_{i \neq j} \beta_{ij}$.

\begin{lemma}\label{lem:3d-u-nu-descent-vol}
  Let $i \in \{1,2,3\}$, $x \in \mathbb{R}e_i$ a 1-sparse vector and $\mu_1, \mu_2, \mu_3 > 0$ three positive real numbers.
  Then,
  \begin{equation}
    \vol\left(\sT_{\|\cdot\|_w}(x) \cap S(1) \right) = 
    4 \tan^{-1} \left( \frac{1}{1 + c_i(\mu)} \right)
  \end{equation}
  where
  \begin{equation}
    c_i(\mu) = 1 + \sum_{j \neq i} \beta_{ij}(\mu) + \prod_{j \neq i} \beta_{ij}(\mu) .
  \end{equation}
\end{lemma}

With this Lemma~\ref{lem:3d-u-nu-descent-vol}, we can prove that the $\ell^1$-norm is the best regularizer for 1-sparse signals both for uniform and non-uniform recovery among all weighted $\ell^1$-norms.

\begin{theorem}\label{th:3d-u-nu} 
  Let $\Sigma= \Sigma_1$ the set of 1-sparse vectors in $\mathbb{R}^3$ and $\sC =\{ \|\cdot\|_w : \|z\|_w = w_1|z_1|+w_2|z_2|+w_3|z_3|, \max |w_i| = 1, w_i  >0 \}$.
  The $\ell^1$-norm is the best regularizer among the class $\sC$ for $\Sigma$, both for the uniform and non-uniform case,
  \begin{equation*}
    \arg \max_{R \in \sC} A_\Sigma^U(R) =
    \arg \max_{R \in \sC} A_\Sigma^{NU}(R) =
    \| \cdot \|_1
  \end{equation*}
\end{theorem}
We expect a similar result if $\sC$ is generalized to asymmetrical atomic norms. To extend these exact calculations to the $n$-dimensional case, we would need to be able to compare intersections of spheres and descent cones of $R$ without an analytic formula, which appears to be a very difficult task.

\section{Compliance measures based on the RIP}\label{sec:RIP_compliance}
In this Section, we detail how we can use RIP recovery conditions to build compliance measures that are more easily managed. We start by recalling definitions and results about RIP recovery guarantees then apply our methodology. We also give Lemma that emphasize the relevant quantity (depending on the geometry of the regularizer and the model) to optimize. 
\begin{definition}[RIP constant] 
 In a Hilbert space $\sH$, let $\Sigma$ a union of subspaces. Let $M$ a linear map, the RIP constant of $M$ is defined as 
 \begin{equation}
  \delta(M) = \inf_{x \in \Sigma-\Sigma} \left| \frac{\|Mx\|_\sH^2}{\|x\|_\sH^2}-1 \right|,
 \end{equation}
 where $\Sigma-\Sigma$ (differences of elements of $\Sigma$) is called the secant set. 
 \end{definition}
In \cite{Traonmilin_2016}, an explicit constant $\delta_\Sigma^{suff}(R)$ is given,  such that $\delta(M)<\delta_\Sigma^{suff}(R)$ guarantees exact recovery of elements of $\Sigma$ by minimization~\eqref{eq:minimization1}. This constant is only sufficient (and sharp in some sense for sparse and low rank recovery). An ideal RIP based compliance measure would be to use a sharp RIP constant $\delta^{sharp}_\Sigma(R)$ (which is not explicit, it is an open question to derive closed form formulations of this constant for sparsity and other low-dimensional models) defined as: 
\begin{equation}
 \delta^{sharp}_\Sigma(R) := \inf_{M : \ker M \cap \sT_R(\Sigma) \neq \{0\} } \delta(M).
\end{equation}
It is the best RIP constant of measurement operators where uniform recovery fail. The lack of analytic expressions for $\delta^{sharp}_\Sigma(R)$ limits the possibilities of exact optimization with respect to $R$.
We propose to look at two compliance measures:  
\begin{itemize}
 \item  Measures based on necessary RIP conditions \cite{Davies_2009} which yields sharp recovery constants for particular set of operators, e.g.,
\begin{equation}
 \delta^{nec}_\Sigma(R) := \inf_{z \in \sT_R(\Sigma) \setminus \{0\}} \delta(I-\Pi_z ).
\end{equation}
 where $\Pi_{z}$ is the orthogonal projection onto the one-dimensional subspace $\tspan(z)$  (other intermediate necessary RIP constants can be defined). Another open question is to determine whether  $\delta_\Sigma^{nec}(R) =  \delta^{sharp}_\Sigma(R)$ generally or for some particular models.
 \item Measures based on sufficient RIP constants for recovery (e.g., $\delta_\Sigma^{suff}(R)$ from \cite{Traonmilin_2016}).
\end{itemize}
Note that we have the relation 
 \begin{equation}
\delta_\Sigma^{suff}(R) \leq  \delta^{sharp}_\Sigma(R)  \leq \delta_\Sigma^{nec}(R). 
\end{equation}
To summarize, in the following, instead of considering the most natural RIP based compliance measure (based on $\delta^{sharp}_\Sigma(R)$ ), we use the best known bounds of this measure.

\subsection{Compliance measures based on necessary RIP conditions}\label{sec:nec_RIP}

In this case, instead of working with actual RIP constants, it is easier to use (equivalently) the restricted conditioning.
\begin{definition}[Restricted conditioning]
\begin{equation}
 \gamma(M) := \frac{\sup_{x\in (\Sigma-\Sigma)\cap S(1)} \|Mx\|_\sH^2}{\inf_{x\in (\Sigma-\Sigma)\cap S(1)} \|Mx\|_\sH^2}.
 \end{equation}
\end{definition}
The RIP constant $\delta(M)$ is increasing with respect to $\gamma(M)$. In the following, we consider the compliance measure 

\begin{equation}
A_\Sigma^{RIP,nec}(R) := \gamma_\Sigma(R) =  \inf_{z \in \sT_R(\Sigma) \setminus \{0\}} \gamma(I-\Pi_z ).
\end{equation}
When maximizing $A_\Sigma^{RIP,nec}(R)$, we look  at optimal regularizers for recovery with tight frames with a kernel of dimension 1. 

From here, we specialize  to $\Sigma = \Sigma_k$  and $\sH = \bR^n$. Hence $\Sigma - \Sigma = \Sigma_{2k}$ with $k\geq1$ and  $n\geq 3$ (for $n<3$ uniform recovery is not possible for non-invertible $M$). To show optimality of the $\ell^1$-norm, we use the following characterization of $A_\Sigma^{RIP,nec}(R)$.

\begin{lemma}\label{le:expr_RC_sparse}
Let  $\Sigma = \Sigma_k$  and $\sH = \bR^n$. Then
\begin{equation}
A_\Sigma^{RIP,nec}(R)= \frac{ 1 }{ 1- \inf_{ z \in \sT_R(\Sigma) \setminus \{0\}} \sup_{x\in (\Sigma-\Sigma)\cap S(1)}\frac{\ls x,z\rs^2}{\|z\|_\sH^2}}.
\end{equation}
\end{lemma}

 We consider the set $\sC = \{ R: R(x) = \sum w_i |x_i| = \|x\|_w, \max |w_i| = 1, w\in (\bR_+^*)^n\}$. Note that 
\begin{equation}
\begin{split}
 \arg \max_{R \in \sC} A_\Sigma^{RIP,nec}(R) = \arg \min_{R \in \sC} \underset{z\in \sT_R(\Sigma) \setminus \{0\} }{\sup}  \frac{\|z_{T_2^c}\|_2^2}{\|z_{T_2}\|_2^2}  =  \arg \min_{R \in \sC} B_\Sigma(R)  \\ 
 \end{split}
\end{equation}
where  $T_2$ is a notation for the support of $2k$ biggest coordinates in $z$, i.e. for all $i \in T_2, j\in T_2^c$, we have $|z_i|\geq |z_j|$. 

Studying the quantity $B_\Sigma(R) := \underset{z\in \sT_R(\Sigma) \setminus \{0\} }{\sup}  \frac{\|z_{T_2^c}\|_2^2}{\|z_{T_2}\|_2^2} $ for $R \in \sC$ permits to conclude.

\begin{theorem} \label{th:RIP_nec}
  Let $\Sigma = \Sigma_k$, $\sH = \bR^n$ and $\sC = \{ R: R(x) = \sum w_i |x_i| = \|x\|_w, \max |w_i| = 1, w\in (\bR_+^*)^n\}$. Suppose  $n\geq 2k$, then
 \begin{equation}
 \|\cdot\|_1 = \arg \max_{R\in \sC} A_\Sigma^{RIP,nec}(R).
 \end{equation}
\end{theorem}
In the next section,  we will see that the optimization of the sufficient RIP constant leads to very similar expressions.

\subsection{Compliance measures based on sufficient RIP conditions }\label{sec:suff_RIP}
In \cite{Traonmilin_2016}, it was shown for $\Sigma$ a union of subspaces and an arbitrary regularizer $R$,  that an explicit RIP constant $\delta_\Sigma^{suff}(R)$ is sufficient to guarantee reconstruction. From here we set $\Sigma := \Sigma_k$, $\sC = \{ R: R(x) = \sum w_i |x_i| = \|x\|_w, \max |w_i| = 1, w\in (\bR_+^*)^n\}$.

\begin{proposition}\label{prop:charac_suff_RIP}
Let $\Sigma = \Sigma_k$ the set of $k$-sparse vectors. Consider the constant  $\delta_\Sigma^{suff}(R)$  from \cite{Traonmilin_2016}[Eq. (5)], we have:
\begin{equation}
\begin{split}
 \delta^{suff}_\Sigma(R) & = \frac{1}{ \sqrt{ \underset{z\in \sT_R(\Sigma) \setminus \{0\} }{\sup}  \frac{\|z_{T^c}\|_\Sigma^2}{\|z_T\|_2^2} + 1 }} \\
 \end{split}
\end{equation}
where $T$ denotes  the  support of the  $k$ biggest coordinates of $z$ and  $\|\cdot\|_\Sigma$ is  the atomic norm generated by $\Sigma \cap S(1)$ (the convex gauge induced by the convex envelope of  $\Sigma \cap S(1)$).
\end{proposition}

In the following, we use 
\begin{equation}\label{eq:compliance_sufficient_RIP}
\begin{split}
 A_\Sigma^{RIP,suff}(R) :=\delta^{suff}_\Sigma(R).
 \end{split}
\end{equation}

It must be noted that  characterization of Proposition~\ref{prop:charac_suff_RIP} was used as a lower bound on $\delta_\Sigma^{suff}$  in \cite{Traonmilin_2016} (hence this part of the proof of the control of $\delta_\Sigma^{suff}$ in \cite{Traonmilin_2016}  is exact).

Similarly to the necessary case, from Proposition~\ref{prop:charac_suff_RIP},  we have 
\begin{equation}
\begin{split}
 \arg \max_{R \in \sC} A_\Sigma^{RIP,suff}(R) = \arg \min_{R \in \sC} \underset{z\in \sT_R(\Sigma) \setminus \{0\} }{\sup}  \frac{\|z_{T^c}\|_\Sigma^2}{\|z_T\|_2^2}  =  \arg \min_{R \in \sC} D_\Sigma(R)  \\ 
 \end{split}
\end{equation}
where $T$ denotes  the  support of the  $k$ biggest coordinates of $z$ and $D_\Sigma(R) := \underset{z\in \sT_R(\Sigma) \setminus \{0\} }{\sup}  \frac{\|z_{T^c}\|_\Sigma^2}{\|z_T\|_2^2} $. Remark the similarity between the fundamental quantity to optimize for the necessary case and the sufficient case $B_\Sigma(R)$ and $D_\Sigma(R)$. Studying the quantity $D_\Sigma(R) $ for $R \in \sC$ leads to the result.

\begin{theorem} \label{th:RIP_suff}
Let $\Sigma = \Sigma_k$, $\sH = \bR^n$ and $\sC = \{ R: R(x) = \sum w_i |x_i| = \|x\|_w, \max |w_i| = 1, w\in (\bR_+^*)^n\}$. Suppose  $n\geq 2k$, then
 \begin{equation}
 \|\cdot\|_1 = \arg \sup_{R\in \sC} A_\Sigma^{RIP,suff}(R).
 \end{equation}
\end{theorem}

\section{Discussion and future work}
We have shown that, not surprisingly, the $\ell^1$-norm is optimal among weighted $\ell^1$-norms for sparse recovery for several notions of compliance. This result had to be expected due to symmetries of the problem. However, the important point is that we could explicitly quantify the notion of good regularizer. This is promising for the search of optimal regularizers for more complicated low-dimensional models such as ``sparse and low rank'' models or hierarchical sparse models. For the case of sparsity, we expect to be able to generalize the optimality of the $\ell^1$-norm to the more general case of atomic norms  with atoms included in the model. We also expect similar result for low-rank recovery and the nuclear norm as technical tools are very similar. 

It must be noted that for RIP compliance measures, we did not use a constructive proof (we exhibited the maximum of the compliance measure). A constructive proof, i.e. an exact calculation and optimization of the quantities $B_\Sigma(R)$ and $D_{\Sigma}(R)$ would be more satisfying as it would not require the knowledge of the optimum, which is our ultimate objective. 

We used compliance measures based on (uniform) RIP recovery guarantees to give results for the general sparse recovery case, it would be interesting to do such analysis using  (non-uniform) recovery guarantees based on the statistical dimension or Gaussian width of the descent cones \cite{Chandrasekaran_2012,Amelunxen_2014}. One would need to precisely lower and upper bound these quantities, similarly to our approach with the RIP, to get satisfying results. 

Finally, while these compliance measures are designed to make sense with respect to known results in the area of sparse recovery, one might design other compliance measures tailored for particular needs, in this search for optimal regularizers.
\section*{Acknowledgements}
This work was partly supported by the CNRS PEPS JC 2018 (project on efficient regularizations). The authors would like to thank Rémi Gribonval for his insights on this problem. 
\section*{References}
\bibliographystyle{abbrv}
 \bibliography{opti_convex_sparse_reg}

\begin{thebibliography}{10}

\bibitem{Amelunxen_2014}
D.~Amelunxen, M.~Lotz, M.~B. McCoy, and J.~A. Tropp.
\newblock Living on the edge: phase transitions in convex programs with random
  data.
\newblock {\em Information and Inference}, 3(3):224--294, 2014.

\bibitem{Argyriou_2012}
A.~Argyriou, R.~Foygel, and N.~Srebro.
\newblock {Sparse Prediction with the k-Support Norm}.
\newblock {\em Advances in Neural Information Processing Systems},
  25:1457--1465, 2012.

\bibitem{Bourrier_2014}
A.~Bourrier, M.~Davies, T.~Peleg, P.~Perez, and R.~Gribonval.
\newblock Fundamental performance limits for ideal decoders in high-dimensional
  linear inverse problems.
\newblock {\em Information Theory, IEEE Transactions on}, 60(12):7928--7946,
  2014.

\bibitem{Cai_2014}
T.~Cai and A.~Zhang.
\newblock Sparse representation of a polytope and recovery of sparse signals
  and low-rank matrices.
\newblock {\em Information Theory, IEEE Transactions on}, 60(1):122--132, 2014.

\bibitem{Chandrasekaran_2012}
V.~Chandrasekaran, B.~Recht, P.~Parrilo, and A.~Willsky.
\newblock The convex geometry of linear inverse problems.
\newblock {\em Foundations of Computational Mathematics}, 12(6):805--849, 2012.

\bibitem{Davies_2009}
M.~E. Davies and R.~Gribonval.
\newblock Restricted isometry constants where $\ell^p$ sparse recovery can fail
  for $0 < p \leq 1$.
\newblock {\em Information Theory, IEEE Transactions on}, 55(5):2203--2214,
  2009.

\bibitem{Foucart_2013}
S.~Foucart and H.~Rauhut.
\newblock {\em A mathematical introduction to compressive sensing}.
\newblock Springer, 2013.

\bibitem{vanoosterom1983solid}
A.~V. Oosterom and J.~Strackee.
\newblock The solid angle of a plane triangle.
\newblock {\em IEEE Transactions on Biomedical Engineering},
  BME-30(2):125--126, 1983.

\bibitem{Puy_2015}
G.~Puy, M.~E. Davies, and R.~Gribonval.
\newblock Recipes for stable linear embeddings from hilbert spaces to
  $\mathbb{R}^m$.
\newblock {\em arXiv preprint arXiv:1509.06947}, 2015.

\bibitem{Soubies_2015}
E.~Soubies, L.~Blanc-F{\'e}raud, and G.~Aubert.
\newblock A continuous exact $\ell_0$ penalty (cel0) for least squares
  regularized problem.
\newblock {\em SIAM Journal on Imaging Sciences}, 8(3):1607--1639, 2015.

\bibitem{Traonmilin_2016}
Y.~Traonmilin and R.~Gribonval.
\newblock {Stable recovery of low-dimensional cones in Hilbert spaces: One RIP
  to rule them all}.
\newblock {\em Applied And Computational Harmonic Analysis, In Press}, 2016.

\bibitem{Vaiter2015Modelselectionlow}
S.~Vaiter, M.~Golbabaee, J.~Fadili, and G.~Peyré.
\newblock Model selection with low complexity priors.
\newblock {\em Information and Inference: A Journal of the IMA}, 4(3):230--287,
  2015.

\end{thebibliography}

\section{Annex}
This section describes the tools and proofs used to obtain our results.
 \subsection{Proofs for Section~\ref{sec:3d_example}}
 
\begin{proof}[Proof of Lemma~\ref{lem:3d-u-nu-descent-vol}]
  We start our proof by observing that the descent cone at $x$ is given by the conic envelope of the rays defined by vectors of the form $\pm \mu_j e_j - \mu_i e_i$, i.e.,
  \begin{equation}
    \sT_{\| \cdot \|_w}(x) = \overline{\cone}
    \left\{
      \pm \mu_j e_j - \mu_i e_i
    \right\}_{j \neq i} .
  \end{equation}
  Using the symmetry of this set, we can split it as four tetrahedra of equal size, depending on the sign in front of the vector $e_j$, 
  \begin{equation}
    \sT_{\| \cdot \|_w}(x) = \bigcup_{s \in \{-1,1\}^2} \overline{\cone}
    \left\{
      (s_j \mu_j e_j - \mu_i e_i)_{j \neq i}, - \mu_i e_i
    \right\}.
  \end{equation}
  We have reduced the problem to computing the area of the intersection between a tetrahedron and the unit sphere.
  To ease the notation, let $i=1$, and we consider the cone
  \begin{equation}
    \mathcal{T} = \overline{\cone}
    \left\{
      \mu_2 e_2 - \mu_1 e_1,
      \mu_3 e_3 - \mu_1 e_1,
      - \mu_1 e_1
    \right\} ,
  \end{equation}
  as depicted in Figure~\ref{fig:3d-1-sparse}.
  Following~\cite[Equation (8)]{vanoosterom1983solid}, we have that
  \begin{equation}
    \tan\left(\frac{1}{2} \vol(\mathcal{T} \cap S(1))\right) =
    \frac{1}{1 + \sum_{j=1}^3 \cos \alpha_{i,j}} ,
  \end{equation}
  Using the definition of $\beta_{ij}$, we have that $\cos \alpha_{ij} = \beta_{ij}$.
  Indeed, using standard relation in a triangle, one has
  \begin{equation*}
    \cos \alpha_{12} = \frac{\mu_1}{\sqrt{\mu_1^2 + \mu_2^2}} = \beta_{12} ,\quad
    \cos \alpha_{13} = \frac{\mu_1}{\sqrt{\mu_1^2 + \mu_3^2}} = \beta_{13} \\
  \end{equation*}
  and, using Al-Kashi relation,
  \begin{equation*}
    \cos \alpha_{11} = \frac{\mu_1^2}{\sqrt{
      \left( \mu_1^2 + \mu_2^2 \right)
      \left( \mu_1^2 + \mu_3^2 \right)
    }} = \beta_{12} \beta_{13} = \beta_{11} .
  \end{equation*}
\end{proof}

\begin{proof}[Proof of Theorem~\ref{th:3d-u-nu}]
  Assume w.l.o.g that $1=w_1 \geq w_2 \geq w_3 > 0$, i.e., $1=\mu_1 \leq \mu_2 \leq \mu_3$.

  \emph{Non-uniform case.}
  Using Lemma~\ref{lem:3d-u-nu-descent-vol}, we compute the quantity
  \begin{align}
    \sup_{x \in \Sigma} \vol \left(\sT_R(x) \cap S(1)\right)
    &= \sup_{i \in \{ 1,2,3 \}} \vol \left(\sT_R(e_i) \cap S(1)\right) \\
    &= \sup_{i \in \{ 1,2,3 \}} 4 \tan^{-1} \left( \frac{1}{1 + c_i(\mu)} \right) . \label{eq:nu-closed-3d-cone}
  \end{align}
  Now, 
  \begin{align*}
    \arg \max_{R \in \sC} A_\Sigma^{NU}(R)
    &= \arg \max_{R \in \sC} \left\{  1 - \sup_{x \in \Sigma} \frac{\vol\left(\sT_R(x) \cap S(1)\right)}{\vol(S(1))} \right\} & \text{by definition~\eqref{eq:anusr}}\\
    &= \arg \min_{R \in \sC} \sup_{x \in \Sigma} \vol\left(\sT_R(x) \cap S(1)\right) \\
    &=  \arg \min_{1=\mu_1 \leq \mu_2 \leq \mu_3} \max_{i \in \{ 1,2,3 \}}  4 \tan^{-1} \left( \frac{1}{1 + c_i(\mu)} \right) & \text{using~\eqref{eq:nu-closed-3d-cone}} \\
    &= \arg \min_{1=\mu_1 \leq \mu_2 \leq \mu_3} \max_{i \in \{ 1,2,3 \}} \frac{1}{1 + c_i(\mu)} & \text{$\tan^{-1}$ is increasing}
  \end{align*}
  Hence, 
  \begin{equation*}
    \arg \max_{R \in \sC} A_\Sigma^{NU}(R) =
    \arg \min_{1=\mu_1 \leq \mu_2 \leq \mu_3} \min_{i \in \{ 1,2,3 \}} c_i(\mu) .
  \end{equation*}
  Observe that under the constraint $1=\mu_1 \leq \mu_2 \leq \mu_3$, the minimum of $c_i(\mu)$ is achieved at $i=3$.
  Indeed, we have $\mu_2^2 \geq \frac{1}{\mu_2^2}$, $\mu_3^2 \geq \frac{1}{\mu_3^2}$, thus $\beta_{12} \geq \beta_{21} \geq \beta_{31}$ and $\beta_{13} \geq \beta_{23} \geq \beta_{32}$.
  Since every term is positive, we also have that $\beta_{12} \beta_{13} \geq \beta_{21} \beta_{23} \geq \beta_{31} \beta_{32}$.
  Thus, $c_1(\mu) \geq c_2(\mu) \geq c_3(\mu)$.
  Hence, we solve the problem
  \begin{equation*}
    \min_{1=\mu_1 \leq \mu_2 \leq \mu_3}
    1 + \left( 1 + \frac{1}{\mu_3^2} \right)^{-1/2}
    + \left( 1 + \frac{\mu_2^2}{\mu_3^2} \right)^{-1/2}
    + \left(  \left(1 + \frac{1}{\mu_3^2}\right)\left(1 + \frac{\mu_2^2}{\mu_3^2}\right) \right)^{-1/2} ,
  \end{equation*}
  which is achieved for $\mu_1 = \mu_2 = \mu_3 = 1$.
  
  \emph{Uniform case.} The computation is similar, replacing the minimum with the sum of $c_i(\mu)$ for $i \in \{ 1,2,3 \}$.
  \begin{align*}
    \vol \left(\sT_R(\Sigma) \cap S(1)\right)
    &= \sum_{i=1}^3 \vol \left(\sT_R(e_i) \cap S(1)\right) & \text{no intersection between the cones} \\
    &= 4 \sum_{i=1}^3 \tan^{-1} \left( \frac{1}{1 + c_i(\mu)} \right) & \text{using Lemma~\ref{lem:3d-u-nu-descent-vol}} .
  \end{align*}
  Again, using the fact than $\tan^{-1}$ and $c_1(\mu) \geq c_2(\mu) \geq c_3(\mu)$, we have that
  \begin{equation*}
    3 \min_{1=\mu_1 \leq \mu_2 \leq \mu_3} \tan^{-1} \left( \frac{1}{1 + c_1(\mu)} \right)
    \leq \min_{1=\mu_1 \leq \mu_2 \leq \mu_3} \sum_{i=1}^3 \tan^{-1} \left( \frac{1}{1 + c_i(\mu)} \right)
    \leq \sum_{i=1}^3 \tan^{-1} \left( \frac{1}{1 + c_i(1,1,1)} \right)
  \end{equation*}
  Since, the left term is minimized by $\mu=(1,1,1)$ as proved in the non-uniform case and that
  \begin{equation*}
    \sum_{i=1}^3 \tan^{-1} \left( \frac{1}{1 + c_i(1,1,1)} \right)
    =  3 \tan^{-1} \left( \frac{1}{1 + c_1(1,1,1)} \right) ,
  \end{equation*}
  we have that the minimum is achieved at $\mu_1 = \mu_2 = \mu_3 = 1$.
  Moreover, since $\mu_1 \mapsto c_1(1,\mu_1,\mu_2)$ and $\mu_2 \mapsto c_1(1,\mu_1,\mu_2)$ are strictly increasing, this minimum is unique.
\end{proof}

\subsection{Proofs for Section~\ref{sec:nec_RIP}}

\subsubsection{Link between RIP constants and restricted conditioning.}\label{sec:link_RIP_RC}
 Suppose $M$ has RIP with constants $0 \leq \delta_1<1$, $\delta_2>-\delta_{1}$ on $\Sigma-\Sigma$, i.e.
\begin{equation}\notag
(1-\delta_1)\|x\|_\sH^2 \leq \|\mM x\|_2^2 \leq (1+\delta_2)\|x\|_\sH^2.
\end{equation}
For any $\lambda>0$ we have 
\begin{equation}\notag
\lambda^2(1-\delta_1)\|x\|_\sH^2 \leq \|\lambda \mM x\|_2^2 \leq \lambda^2(1+\delta_2)\|x\|_\sH^2
\end{equation}
We look for $\lambda,\delta'$ such that $\lambda M$ has symmetric RIP $\delta'$:
\begin{equation}
\begin{split}
 \lambda^2(1-\delta_1) &= 1-\delta'\\
 \lambda^2(1+\delta_2) &= 1+\delta'
\end{split}
\end{equation}
It implies 
\begin{equation}
\begin{split}
 \frac{(1-\delta')}{1-\delta_1}(1+\delta_2) &= 1+\delta' \\
 \delta' (1+ \frac{1+\delta_2}{1-\delta_1}) &=  \frac{(1+\delta_2)}{1-\delta_1}- 1  \\   
 \delta'  \frac{2+\delta_2-\delta_1}{1-\delta_1}& = \frac{\delta_2+\delta_1}{1-\delta_1}\\
 \delta' & = \frac{\delta_2+\delta_1}{2+\delta_2-\delta_1} 
\end{split}
\end{equation}

 We can always rescale $M$ such that $M$ has RIP $\delta'$. taking $\delta_1(M)= 1-\inf_{x\in (\Sigma-\Sigma)\cap S(1)} \|Mx\|_\sH^2$ and 
 $\delta_2(M)= \sup_{x\in (\Sigma-\Sigma)\cap S(1)} \|Mx\|_\sH^2-1$ 
We have 
$\delta' = \frac{\sup_{x\in (\Sigma-\Sigma)\cap S(1)} \|Mx\|_\sH^2 -\inf_{x\in (\Sigma-\Sigma)\cap S(1)} \|Mx\|_\sH^2}{ \sup_{x\in (\Sigma-\Sigma)\cap S(1)} \|Mx\|_\sH^2 +\inf_{x\in (\Sigma-\Sigma)\cap S(1)} \|Mx\|_\sH^2} = \frac{\gamma(M)-1}{\gamma(M)+1}$.

\subsubsection{Proofs}
\begin{proof}[Proof of Lemma~\ref{le:expr_RC_sparse}]
Let $z \in \sH$. Let $M =I - \Pi_z$ where $\Pi_{z}$ is the orthogonal projection onto the one-dimensional subspace $\tspan(z)$.

Let $x \in \Sigma - \Sigma$, we always have $\|Mx\|_2 \leq \|x\|_\sH$ and 
\begin{equation}\notag
 \|Mx\|_2^2 = \|x-\Pi_z x\|_\sH^2 = \|x\|_\sH^2 -2 \ls x, \Pi_z x\rs + \|\Pi_z x\|_\sH^2=\|x\|_\sH^2 - \frac{\ls x,z\rs^2}{\|z\|_\sH^2}.
\end{equation} 
Thus 
\begin{equation}
 \frac{\|Mx\|_2^2 }{\|x\|_\sH^2} = 1 -\frac{\ls x,v\rs^2}{\|x\|_\sH^2\|v\|_\sH^2}
\end{equation}
and 
\begin{equation}
\begin{split}
\gamma(M) &= \frac{\sup_{x\in (\Sigma-\Sigma)\cap S(1)} \|Mx\|_\sH^2}{\inf_{x\in (\Sigma-\Sigma)\cap S(1)} \|Mx\|_\sH^2} \\
&= \frac{ 1 + \sup_{x\in (\Sigma-\Sigma)\cap S(1)} -\frac{\ls x,z\rs^2}{\|z\|_\sH^2} }{ 1+ \inf_{x\in (\Sigma-\Sigma)\cap S(1)} -\frac{\ls x,z\rs^2}{\|z\|_\sH^2}}. \\
&= \frac{ 1  -\inf_{x\in (\Sigma-\Sigma)\cap S(1)} \frac{\ls x,z\rs^2}{\|z\|_\sH^2} }{ 1-\sup_{x\in (\Sigma-\Sigma)\cap S(1)}\frac{\ls x,z\rs^2}{\|z\|_\sH^2}} 
\end{split}
\end{equation}
For all $z \in \sH$, there is $x \in \Sigma - \Sigma =\Sigma_{2k}$such that  $\ls x,z\rs =0$, hence $\inf_{x\in (\Sigma-\Sigma)\cap S(1)} \frac{\ls x,z\rs^2}{\|z\|_\sH^2}=0$  (just take $x \perp z_T$ where $T$ is the  support of the greatest $2k$ coordinates of $z$) . This gives  the desired result.
\end{proof}

The following  Lemma shows that we can simplify the calculation of $B_\Sigma(R)$

\begin{lemma} \label{lem:optimal_support_l2}
Let $\Sigma = \Sigma_k$ the set of $k$-sparse vectors. Let $R =\|\cdot\|_w  \in \sC$,
\begin{equation}
B_\Sigma(R) := \sup_{ z \in \sT_R(\Sigma) \setminus \{0\} } \frac{\|z_{T_2^c}\|_2^2}{\|z_{T_2}\|_2^2}  =  \sup_{ z \neq 0 :\|z_{H_0}\|_w\geq\|z_{H_0^c}\|_w} \frac{\|z_{T_2^c}\|_2^2}{\|z_{T_2}\|_2^2}.\\
\end{equation}
where $H_0$ is the support of the $k$ biggest weights: i.e. for all $i \in H_0,j \notin H_0$, we have $w_i \geq w_j$. 
\end{lemma}
\begin{proof}
 
We have 
\begin{equation}
B_\Sigma(R) = \sup_{H : |H|\leq k}\sup_{z \neq 0 : \|z_H\|_w\geq\|z_{H^c}\|_w} \frac{\|z_{T_2^c}\|_2^2}{\|z_{T_2}\|_2^2}.  \\
\end{equation}
Let $z$ such that $\|z_H\|_w\geq\|z_{H^c}\|_w$, there is $z'$ a permutation of $z$ (adequately transfer the values of $H$ to $H_0$ and the values on $H^c$ to $H_0^c$) such that $\|z_{H_0}'\|_w \geq \|z_H\|_w\geq\|z_{H^c}\|_w \geq \|z_{H_0^c}'\|_w $. We remark that the expression $ \frac{\|z_{T_2^c}\|_2^2}{\|z_{T_2}\|_2^2} $ does not depend on the order of the coordinates of $z$. This  shows that $\sup_{H : \|z_H\|_w\geq\|z_{H^c}\|_w} \frac{\|z_{T_2^c}\|_2^2}{\|z_{T_2}\|_2^2} \leq \sup_{H_0 : \|z_{H_0}\|_w\geq\|z_{H_0^c}\|_w}  \frac{\|z_{T_2^c}\|_2^2}{\|z_{T_2}\|_2^2} $ which concludes the proof.
\end{proof}

\begin{lemma} \label{lem:charact_supBp_final}
Let $\Sigma = \Sigma_k$ the set of $k$-sparse vectors. Let $R =\|\cdot\|_w  \in \sC$ such that $ \|w_{H_0}\|_1< \|w_{H_0^c}\|_1$. We have 
\begin{equation}
B_\Sigma(R) = \sup_{ z \in \sT_R(\Sigma) \setminus \{0\} } \frac{\|z_{T_2^c}\|_2^2}{\|z_{T_2}\|_2^2}  =  \sup_{z \neq 0 : \|z_{H_0}\|_w=\|z_{H_0^c}\|_w} \frac{\|z_{T_2^c}\|_2^2}{\|z_{T_2}\|_2^2}  \\
\end{equation}
where $H_0$ is the support of the $k$ biggest weights.
\end{lemma}
\begin{proof}
 We use Lemma~\ref{lem:optimal_support_l2}. Let $z$ such that $\|z_{H_0}\|_w-\|z_{H_0^c}\|_w=\beta>0$. This implies that $z$ is not a constant vector otherwise the hypothesis on weights would be violated.  If there is $i_0 \in T_2^c$ such that $|z_{i_0}| <\min_{i\in T_2} |z_j|=\alpha $ we set $z'$ such that $|z_{i_0}'| = \min(\alpha, |z_{i_0}| + |	\beta|)  $  otherwise  there is  $i_0 \in T_2$ such that $|z_{i_0}| >\min_{i\in T_2} |z_j|=\alpha $, we set $z'$ such that $|z_i'| = \max(\alpha, |z_{i_0}| - \beta)  $ and $|z_i'| = |z_i|$ for $i \neq i_0$. In both case, we have 
 \begin{equation}
  \frac{\|z_{T_2^c}\|_2^2}{\|z_{T_2}\|_2^2}  \leq \frac{\|z_{T_2'^c}\|_2^2}{\|z_{T_2'}\|_2^2} .
 \end{equation}
We either have $\|z_{H_0}'\|_w=\|z_{H_0^c}'\|_w$ or $\|z_{H_0}'\|_w-\|z_{H_0^c}'\|_w >\beta'$, in this last case we strictly reduced the size of the support of $z_i$ such that  $z_i \neq \min_{i\in T_2} |z_j|$. Induction on the size of this support leads to the conclusion that there is always a $z'$ such that $\|z_{H_0}'\|_w=\|z_{H_0^c}'\|_w$  and $\frac{\|z_{T_2^c}\|_2^2}{\|z_{T_2}\|_2^2}  \leq \frac{\|z_{T_2'^c}\|_2^2}{\|z_{T_2'}\|_2^2} $.
\end{proof}

This Lemma allows to bound $B_\Sigma(R)$ when $R$ is not to far from the $\ell^1$-norm

\begin{lemma}\label{lem:charact_supBL2}
Let $\Sigma = \Sigma_k$ the set of $k$-sparse vectors. Let $R =\|\cdot\|_w \in \sC$. Let $H_0$  a support of $k$ biggest weights in $w$, let $L\geq 1$ an integer and $H_1$ a support of  $k+L$ smallest weights in $w$. Suppose $\|w_{H_0}\|_1 < \|w_{H_1}\|_1 $, then
\begin{equation}
B_L(R) := \sup_{ z \neq 0 : \|z_{H_0}\|_w=\|z_{H_0^c}\|_w , |supp(z)| =2k+L} \frac{\|z_{T_2^c}\|_2^2}{\|z_{T_2}\|_2^2}  \geq \frac{  \frac{L}{k}}{ \frac{\|w_{H_1}\|_1^2}{\|w_{H_0}\|_1^2} + 1 }. \\
\end{equation}
\end{lemma}

\begin{proof}
Let $z$ such that for $j \in H_1$, $|z_j| = \alpha$, for $j\in H_0$, $|z_j| = \beta$ and $z_j=0$ otherwise. Set  $0 \neq \beta = \frac{\|w_{H_1}\|_1}{\|w_{H_0}\|_1}\alpha $. Hence, $\|z_{H_0}\|_w=\|z_{H_1}\|_w$ and, using the hypothesis, $\beta \geq \alpha$. This gives

\begin{equation}
  \frac{\|z_{T_2^c}\|_2^2}{\|z_{T_2}\|_2^2} = \frac{  L\alpha^2}{ k\beta^2 + k\alpha^2 } =\frac{  L}{ k\frac{\|w_{H_1}\|_1^2}{\|w_{H_0}\|_1^2} + k }
 \end{equation}
 which is the desired result.
\end{proof}

We can conclude with the proof of the theorem.
\begin{proof}[Proof of Theorem~\ref{th:RIP_nec}]
 Let $\sC_L =\{\|\cdot\|_w \in \sC : \|w_{H_0}\|_1 < \|w_{H_1}\|_1 \}$  (we define  $H_0$  as the support of $k$ biggest weights in $w$, and  $H_1$ as the support of the $k+L$ smallest weights).  We have $\frac{\|w_{H_1}\|_1}{\|w_{H_0}\|_1} \leq \frac{L+k}{k}$. Hence, for $R \in \sC_L$, with Lemma~\ref{lem:charact_supBL2}, the value $B_L(R)$ is lower bounded (and reached) by the case $\|\cdot\|_1$:
\begin{equation}\label{eq:ineq_BL1}
\begin{split}
 B_L(R)\geq \frac{  \frac{L}{k}}{ \frac{\|w_{H_1}\|_1^2}{\|w_{H_0}\|_1^2} + 1 } \geq B_L(\|\cdot\|_1) =\frac{ \frac{L}{k}}{ \left(\frac{L}{k}+1\right)^2 + 1 }, \\
 \end{split}
\end{equation}
where the value of $B_L(\|\cdot\|_1)$ was obtained by calculating $\frac{\|z_{T_2^c}\|_2^2}{\|z_{T_2}\|_2^2}$ for the optimal $z$ given by \cite{Davies_2009} ($z$ flat on $H_0$ and flat on $H_1$).  The maximum value of $B_L(\|\cdot\|_1)$  (with respect to $L$) is reached for $L/k$ maximizing $f(u) = u/((u+1)^2+1)$ (which is maximized at $\sqrt{2}$ over $\bR$, see Lemma~\ref{lem:function_study1}). We verify that it matches the necessary RIP condition $\frac{1}{\sqrt{2}}$ from \cite{Davies_2009} ($f(\sqrt{2})= 2\sqrt{2}/(2+\sqrt{2})$ which gives $\gamma_\Sigma(\|\cdot\|_1) = (4+3\sqrt{2})/\sqrt{2} = \frac{\sqrt{2}+1}{\sqrt{2}-1}$). Note that the inequality~\eqref{eq:ineq_BL1} is strict if $R \neq \|\cdot\|_1$.

If $R \in \sC_L'$ the complementary in $\sC$ of $\sC_L $, take $z$ such that $z_i =\alpha$ for $i \in H_0 \cup H_1$, $z_i=0$ otherwise. We have $\|z_{H_0}\|_w \geq \|z_{H_1}\|_w$ and
\begin{equation}
B_L(R) \geq \frac{L}{2k} > \frac{ \frac{L}{k}}{ \left(\frac{L}{k}+1\right)^2 + 1 } =   B_L(\|\cdot\|_1)
\end{equation}
Hence, by noting $L_0 = \arg \max_L B_L(\|\cdot\|_1)$, we just showed that for any $R \in \sC$, 
\begin{equation}
  B_\Sigma(R)=\max_L B_L(R)  \geq B_{L_0}(R) \geq \max_{L}  B_L(\|\cdot\|_1) =  B_\Sigma(\|\cdot\|_1)
\end{equation}
with equality only when $R = \|\cdot\|_1$. Taking the $\inf$ with respect to $R\in \sC$ of both sides of this inequality yields the result.
\end{proof}

\begin{lemma} \label{lem:function_study1}
The function $f(u)= \frac{u}{(u+1)^2+1}$ is increasing for $0\leq  u\leq \sqrt{2}$ and decreasing for $ \sqrt{2}\leq u\leq +\infty$
\end{lemma}
\begin{proof}
 For $u \geq 0$,  $f'(u) = 0$ is equivalent to $(u+1)^2+1 -2u(u+1) =0$ which yields $u=\sqrt{2}$. Using the fact that $f(0)=0$ and $\lim_{u\to\infty} f(u) =0$ gives the result.
\end{proof}

\subsection{Proofs for Section~\ref{sec:suff_RIP}}
We refer the reader to  \cite{Traonmilin_2016}[Section 2.2] and \cite{Argyriou_2012} for properties of the atomic norm $\|\cdot\|_\Sigma$
\begin{proof}[Proof of \ref{prop:charac_suff_RIP}]
 The constant $\delta_\Sigma^{suff}(R)$ \cite{Traonmilin_2016}[Eq. (5)] has the following expression:
\begin{equation}
 \delta_\Sigma^{suff}(R)  = \underset{z\in \sT_R(\Sigma) \setminus \{0\} }{\inf}  \underset{x\in \Sigma}{\sup} \frac{-\re \ls x,z \rs}{\|x\|_\sH \sqrt{ \|x+z\|_\Sigma^2 -  \|x\|_\sH^2 - 2\re \ls x,z \rs}}
\end{equation}
 
 We show that for any $z$, 
 \begin{equation}
\begin{split}
  \underset{x\in \Sigma}{\sup} \frac{-\re \ls x,z \rs}{\|x\|_\sH \sqrt{ \|x+z\|_\Sigma^2 -  \|x\|_\sH^2 - 2\re \ls x,z \rs}} =  \frac{1}{ \sqrt{ \underset{z\in \sT_R(\Sigma) \setminus \{0\} }{\sup}  \frac{\|z_{T^c}\|_\Sigma^2}{\|z_T\|_2^2} + 1 }}. \\
 \end{split}
\end{equation}
Taking $x = -z_T$, immediately shows that 

 \begin{equation}
\begin{split}
  \underset{x\in \Sigma}{\sup} \frac{-\re \ls x,z \rs}{\|x\|_\sH \sqrt{ \|x+z\|_\Sigma^2 -  \|x\|_\sH^2 - 2\re \ls x,z \rs}} \geq  \frac{1}{ \sqrt{ \underset{z\in \sT_R(\Sigma) \setminus \{0\} }{\sup}  \frac{\|z_{T^c}\|_\Sigma^2}{\|z_T\|_2^2} + 1 }}. \\
 \end{split}
\end{equation}
Let $x$ such that $supp(x) = H$. Let $u_i\in \Sigma,\lambda_i\geq0$ such that $\sum_i \lambda_i =1$, $\sum_i \lambda_i u_i = x+z$ and $\|x+z\|_\Sigma^2 = \sum \lambda_i \|u_i\|^2$ (from~\cite{Traonmilin_2016}[Fact 2.1]  $\|v|_\Sigma$ is the infimum of this expression over the convex  decompositions $\lambda_i,u_i$ of $v$). Using the fact that $\|v\|_\Sigma^2 =\|v\|_2^2$ for $v \in \Sigma$, we have 
\begin{equation}
\begin{split}
 \|z_{H^c}\|_\Sigma^2  +  \|x_H+z_H\|_2^2 &= \|z_{H^c}\|_\Sigma^2  +  \|x_H+z_H\|_\Sigma^2 \leq \sum_i \lambda_i \|u_{i,H^c}\|_2^2 + \|\sum_i\lambda_i u_{i,H}\|_2^2\\
  &\leq \sum_i \lambda_i \|u_{i,H^c}\|_2^2 + \sum_i\lambda_i \| u_{i,H}\|_2^2\\
   &= \sum_i \lambda_i \|u_{i}\|_2^2  = \|x+z\|_\Sigma^2.  \\
 \end{split}
\end{equation}
We have 

\begin{equation}
\begin{split}
 \|z_{H^c}\|_\Sigma^2  +  \|x_H+z_H\|_2^2 &\leq \|x+z\|_\Sigma^2 \\
  \|z_{H^c}\|_\Sigma^2  +   \|z_H\|_2^2  &\leq \|x+z\|_\Sigma^2 - \|x_H+z_H\|_2^2 + \|z_H\|_2^2\\
   \|z_{H^c}\|_\Sigma^2  +   \|z_H\|_2^2  &\leq \|x+z\|_\Sigma^2 - \|x_H\|_2^2 -2\re \ls x_H, z_H \rs \\
   \|z_{H^c}\|_\Sigma^2  +   \|z_H\|_2^2  &\leq \|x+z\|_\Sigma^2 - \|x_H\|_2^2 -2\re \ls x_H, z\rs \\
 \end{split}
\end{equation}
Using Cauchy-Schwartz inequality, we get 

\begin{equation}
 \begin{split}
  \re \ls x_H, z\rs^2  \left(\|z_{H^c}\|_\Sigma^2  +   \|z_H\|_2^2 \right) &\leq \|x_H\|_2^2\|z_H\|_2^2\left(\|x+z\|_\Sigma^2 - \|x_H\|_2^2 -2\re \ls x_H, z\rs\right) \\
  \frac{\re \ls x_H, z\rs^2}{ \|x_H\|_2^2\left(\|x+z\|_\Sigma^2 - \|x_H\|_2^2 -2\re \ls x_H, z\rs\right) }  &\leq\frac{\|z_H\|_2^2}{\left(\|z_{H^c}\|_\Sigma^2  +   \|z_H\|_2^2 \right)} =  \frac{1}{\frac{z_{H^c}}{\|z_{H}\|_2^2}  +1} \leq \frac{1}{\frac{\|z_{T^c}\|_\Sigma^2}{\|z_{T}\|_2^2}  +1}. \\
  \end{split}
 \end{equation}
where the last inequality comes from the fact that $\|z_{H}\|_2^2 \leq \|z_{T}\|_2^2 =\|z_{T}\|_\Sigma^2$ and $\|z_{T^c}\|_\Sigma^2 \leq \|z_{H^c}\|_\Sigma^2$ (see Lemma~\ref{lem:sigma_norm_increasing2}). 
 
This shows that 

 \begin{equation}
\begin{split}
  \underset{x\in \Sigma}{\sup} \frac{-\re \ls x,z \rs}{\|x\|_\sH \sqrt{ \|x+z\|_\Sigma^2 -  \|x\|_\sH^2 - 2\re \ls x,z \rs}} \leq  \frac{1}{ \sqrt{ \underset{z\in \sT_R(\Sigma) \setminus \{0\} }{\sup}  \frac{\|z_{T^c}\|_\Sigma^2}{\|z_T\|_2^2} + 1 }}. \\
 \end{split}
\end{equation}
\end{proof}

\begin{lemma}\label{lem:const_max_sigmaw_l1} We have 
 \begin{equation}
   \lambda^2 = D_1 := \sup \|z\|_\Sigma^2 \; s.t.  \; \|z\|_1 = \lambda.
 \end{equation}
 The sup is reached for any $z^*$ such that $|supp(z^*)|=1$.
\end{lemma}
\begin{proof}
 Let an admissible $z^*$ such that $|supp(z^*)|=1 $. We have $\|z^*\|_\Sigma^2 = \|z^*\|_2^2 =\lambda^2 $ and $\|z^*\|_1 = \lambda$ hence $D_1 \geq  \lambda^2 $. Moreover, for any $z$ satisfying $\|z\|_1 = \lambda$, $\|z\|_\Sigma^2 \leq \|z\|_1^2   = \lambda^2$.
\end{proof}

To replicate the proof used in the necessary case, we show a monotony property of $\|\cdot\|_\Sigma$.
\begin{lemma}\label{lem:sigma_norm_increasing}
let $v,v' \in \sH$  such that there is $j_0$ such that $|v_{j_0}|\leq |v_{j_0}'| $ and for $j \neq j_0$, $|v_{j}|= |v_{j}'| $. Then $\|v\|_\Sigma \leq\|v'\|_\Sigma$.
\end{lemma}
\begin{proof}
 Let $\lambda_i,u_i$ such that $u_i'\in \Sigma,\lambda_i$ such that $\sum_i \lambda_i =1$, $\sum_i \lambda_i u_i' = v'$ and $\|v'\|_\Sigma^2 = \sum \lambda_i \|u_i'\|^2$ (from~\cite{Traonmilin_2016}[Fact 2.1]  $\|v'\|_\Sigma$ is the infimum of this expression over such convex  decompositions $\lambda_i,u_i'$ of $v'$). Let $H = supp(v)\setminus j_0$.
 
 Leveraging \cite{Traonmilin_2016}[Fact 2.1] for $v$, we have 
 \begin{equation}
 \begin{split}
  \|v\|_\Sigma^2 &\leq \sum \lambda_i \|u_{i,H}'\|_2^2+ \sum \lambda_i \|v_{j_0}\|_2^2 \\
    &\leq  \sum \lambda_i \|u_{i,H}'\|_2^2+ \sum \lambda_i \|v'_{j_0}\|_2^2 \\
  &\leq  \sum \lambda_i \|u_{i,H}'\|_2^2+ \sum \lambda_i \|u'_{i,j_0}\|_2^2 \\
  &= \sum \lambda_i \|u_i'\|_2^2 = \|v'\|_\Sigma^2.
  \end{split}
 \end{equation}
\end{proof}

A consequence of this Lemma is 
\begin{lemma}\label{lem:sigma_norm_increasing2}
Let $v,v'$  such that for all $j $, $|v_{j}|\leq |v_{j}'| $. Then $\|v\|_\Sigma \leq\|v'\|_\Sigma$.
\end{lemma}
\begin{proof}
Consider for $r=0,n$ (We work in $\bR^n$) a vector $w^r$  such that $w^0 = v'$, $w^r_i = v_i$ for $ i= 1,r$ and  $w^r(i) = v_i'$ otherwise. With Lemma~\ref{lem:sigma_norm_increasing} we have $\|v'\|_\Sigma =\|w^0\|_\Sigma \geq \|w^1\|_\Sigma...  \geq\|w^n\|_\Sigma=\|v\|_\Sigma $. 

\end{proof}

\begin{lemma} \label{lem:optimal_support_nsigma}
Let $\Sigma = \Sigma_k$ the set of $k$-sparse vectors. Let $R = \|\cdot\|_w \in \sC$. We have
\begin{equation}
D'(R) := \sup_{ z \in \sT_R(\Sigma) \setminus \{0\} } \frac{\|z_{T^c}\|_\Sigma^2}{\|z_{T}\|_2^2}  =  \sup_{ z \neq 0 :\|z_{H_0}\|_w\geq\|z_{H_0^c}\|_w} \frac{\|z_{T^c}\|_\Sigma^2}{\|z_{T}\|_2^2}.\\
\end{equation}
where $H_0$ is the support of the $k$ biggest weights: for all $i \in H_0,j \notin H_0$ $w_i \geq w_j$. 
\end{lemma}
\begin{proof}
Using the fact that $\|z\|_\Sigma$ is invariant by permutation of coordinates $z$, the proof of Lemma~\ref{lem:optimal_support_l2} is valid in this case.
\end{proof}

\begin{lemma} \label{lem:charact_supDp_final}
Let $\Sigma = \Sigma_k$ the set of $k$-sparse vectors. Let $R =\|\cdot\|_w \in \sC$ such that $\|w_{H_0}\|_1<\|w_{H_0^c}\|_1$. We have 
\begin{equation}
D_\Sigma(R) = \sup_{ z \in \sT_R(\Sigma) \setminus \{0\} } \frac{\|z_{T^c}\|_\Sigma^2}{\|z_{T}\|_2^2}  =  \sup_{z \neq 0 : \|z_{H_0}\|_w=\|z_{H_0^c}\|_w} \frac{\|z_{T^c}\|_\Sigma^2}{\|z_{T}\|_2^2}. \\
\end{equation}
where $H_0$ is the support of the $k$ biggest weights.
\end{lemma}
\begin{proof}
Using Lemma~\ref{lem:sigma_norm_increasing} the proof of Lemma~\ref{lem:charact_supBp_final} holds in this case.
\end{proof}

We calculate  tools to obtain $D_\Sigma(\|\cdot\|_1)$ the case of the $\ell^1$-norm.
\begin{lemma}\label{lem:sup_DL_l1}
\begin{equation}
D_L(\|\cdot\|_1) :=  \sup_{ z \neq 0 : \|z_{H_0}\|_1=\|z_{H_0^c}\|_1 , |supp(z)| =k+L}  \frac{\|z_{T^c}\|_\Sigma^2}{\|z_{T}\|_2^2}   = \min \left(1,\frac{L}{k}\right)\\
\end{equation}
\end{lemma}
\begin{proof}
It was already proven in \cite{Traonmilin_2016}[Theorem 4.1]that  
  \begin{equation}
   \sup_{ z \neq 0 : \|z_{H_0}\|_1=\|z_{H_0^c}\|_1 , |supp(z)| =2k+L}  \frac{\|z_{T^c}\|_\Sigma^2}{\|z_{T}\|_2^2} \leq 1
 \end{equation}
 
 If $L \geq k$, Take $z$ such that $z_{H_0}= ( \beta,..,\beta) $ and $z_{H_0^c \cap supp(z)}= ( \alpha,..,\alpha) $ with  $\beta = \frac{L}{k}\alpha \geq \alpha$, we have $T = H_0$.

From \cite{Argyriou_2012}[Proposition 2.1], $\|z_{T^c}\|_\Sigma^2 = \frac{1}{k}(\|z_{H_0^c \cap supp(z)}\|_1)^2 = L^2\alpha^2/k$. Moreover $\|z_{T}\|_2^2 =k\beta^2 = L^2\alpha^2/k $, thus
 \begin{equation}
    \frac{\|z_{T^c}\|_\Sigma^2}{\|z_{T}\|_2^2} = 1.
 \end{equation}
 
 If $ L<k$, we have $\|z_{T^c}\|_\Sigma^2 = \|z_{T^c}\|_2^2$ and
 \begin{equation}
    \frac{\|z_{T^c}\|_\Sigma^2}{\|z_{T}\|_2^2} \leq L/k.
 \end{equation}
 taking $z = (1,1,1..1,0,0,0)$ shows that this bound is reached.

\end{proof}
We bound $D_\Sigma(R)$ for norms not too far from the $\ell^1$-norm.

\begin{lemma}\label{lem:charact_supDL2}
Let $R =\|\cdot\|_w \in \sC$. Let $H_0$  the support of $k$ biggest weights in $w$, let $L\geq1$ an integer and $H_1$ a  support of  $L$ smallest weights in $w$. Suppose $\|w_{H_0}\|_1< \|w_{H_1}\|_1$

Then, for $L\geq k$
\begin{equation}
D_L(R) = \sup_{ z \neq 0 : \|z_{H_0}\|_w=\|z_{H_0^c}\|_w , |supp(z)| =k+L} \frac{\|z_{T^c}\|_\Sigma^2}{\|z_{T}\|_2^2}  \geq \frac{ \frac{1}{k}L^2}{ k\frac{\|w_{H_1}\|_1^2}{\|w_{H_0}\|_1^2} }. \\
\end{equation}
For $L\leq k$
\begin{equation}
D_L(R) = \sup_{ z \neq 0 : \|z_{H_0}\|_w=\|z_{H_0^c}\|_w , |supp(z)| =k+L} \frac{\|z_{T^c}\|_\Sigma^2}{\|z_{T}\|_2^2}  \geq \frac{L}{k}. \\
\end{equation}
\end{lemma}

\begin{proof}
Let $z$ such that for $j \in H_1$, $|z_j| = \alpha$, for $j\in H_0$, $|z_j| = \beta$ and $z_j=0$ otherwise. Set  $0 \neq \beta = \frac{\|w_{H_1}\|_1}{\|w_{H_0}\|_1}\alpha $. Hence, $\|z_{H_0}\|_w=\|z_{H_1}\|_w$ and, using the hypothesis, $\beta \geq \alpha$. For $L\geq k$, this gives

\begin{equation}
  \frac{\|z_{T^c}\|_\Sigma^2}{\|z_{T}\|_2^2} = \frac{  \|y\|_\Sigma^2\alpha^2}{ k\beta^2 } =\frac{  \|y\|_\Sigma^2}{ k\frac{\|w_{H_1}\|_1^2}{\|w_{H_0}\|_1^2} }
 \end{equation}
 where $y_i =1$ for $i\in H_1$ and $y_i=0$ otherwise. From \cite{Argyriou_2012}[Proposition 2.1] $\|y\|_\Sigma^2 = \frac{1}{k}L^2$ 
 
\begin{equation}
 \frac{\|z_{T^c}\|_\Sigma^2}{\|z_{T}\|_2^2} =\frac{ \frac{1}{k}L^2}{ k\frac{\|w_{H_1}\|_1^2}{\|w_{H_0}\|_1^2} }. 
 \end{equation}
 For $L \leq K$, take $z = (1,1,1..1,0,0,0)$, using the hypothesis on $H_1$ and $H_0$, we have $\|z_{H_0}\|_w\geq \|z_{H_1}\|_w$ and
\begin{equation}
 \frac{\|z_{T^c}\|_\Sigma^2}{\|z_{T}\|_2^2} = \frac{L}{k} . 
 \end{equation}
\end{proof}

\begin{proof}[Proof of Theorem~\ref{th:RIP_suff}]
 For $L > k$, let $\sC_L =\{\|\cdot\|_w \in \sC : \|w_{H_0}\|_1 <\|w_{H_1}\|_1 \}$  (we define  $H_0$  as the support of $k$ biggest weights in $w$, and  $H_1$ as the support of the $L$ smallest weights). We have $\frac{\|w_{H_1}\|_1}{\|w_{H_0}\|_1} \leq \frac{L}{k}$. Hence, for $R \in \sC_L$, with Lemma~\ref{lem:charact_supDL2}, the value $D_L(R)$ is lower bounded (and reached) by the case $\|\cdot\|_1$ (Lemma~\ref{lem:sup_DL_l1}):
\begin{equation}\label{eq:ineq_DBL1}
\begin{split}
 D_L(R) \geq \min \left(1, \frac{L}{k}\right)  =D_L(\|\cdot\|_1)  \\
 \end{split}
\end{equation}
 We verify that it matches the sufficient RIP condition of $\frac{1}{\sqrt{2}}$ for $\ell^1$ recovery from \cite{Cai_2014,Traonmilin_2016}.

If $R \in \sC_L'$ the complementary in $\sC$ of $\sC_L $, take $z$ such that $z_i =\alpha$ for $i \in H_0 \cup H_1$, $z_i=0$ otherwise. We have $\|z_{H_0}\|_w \geq \|z_{H_1}\|_w$ and
\begin{equation}
D_L(R) \geq \frac{L}{k} \geq \min \left(1, \frac{L}{k}\right)  =   D_L(\|\cdot\|_1)
\end{equation}
Hence, we just showed that for any $R \in \sC$,
\begin{equation}
  D_\Sigma(R)=\max_L D_L(R)  \geq 1=\max_{L}  D_L(\|\cdot\|_1) =  D_\Sigma(\|\cdot\|_1)
\end{equation}
with equality only when $R = \|\cdot\|_1$. Taking the $\inf$ with respect to $R\in \sC$ of both sides of this inequality yields the result.
\end{proof}

\end{document}